\documentclass[11pt,a4paper,UKenglish,cleveref, autoref, thm-restate]{article}
\usepackage[font=footnotesize]{caption}
\usepackage{comment}
\usepackage{amsmath}
\usepackage{amsthm}
\usepackage{amssymb}
\usepackage{url}
\usepackage[final]{showkeys}
\usepackage{xcolor}

\usepackage{fullpage}

\usepackage{algorithmic}
\usepackage{floatpag}

\newtheorem{theorem}{Theorem}
\newtheorem{lemma}{Lemma}

\newtheorem{corollary}{Corollary}
\newtheorem{definition}{Definition}

\newtheorem*{open*}{Open Problem}

\usepackage[noend]{algorithm2e}
\newcommand\restr[2]{{% we make the whole thing an ordinary symbol
  \left.\kern-\nulldelimiterspace % automatically resize the bar with \right
  #1 % the function
  \vphantom{\big|} % pretend it's a little taller at normal size
  \right|_{#2} % this is the delimiter
  }}
\usepackage{float}
\usepackage{amssymb}
\newcommand{\E}{\mathbb{E}}

\begin{document}

\title{Faster algorithm for Unique $(k,2)$-CSP} 
\author{Or Zamir \\ Institute for Advanced Study, Princeton, NJ USA}
\date{}
\maketitle

\begin{abstract}
In a $(k,2)$-Constraint Satisfaction Problem we are given a set of arbitrary constraints on pairs of $k$-ary variables, and are asked to find an assignment of values to these variables such that all constraints are satisfied. 
The $(k,2)$-CSP problem generalizes problems like $k$-coloring and $k$-list-coloring.
In the Unique $(k,2)$-CSP problem, we add the assumption that the input set of constraints has at most one satisfying assignment.

Beigel and Eppstein gave an algorithm for $(k,2)$-CSP running in time $O\left(\left(0.4518k\right)^n\right)$ for $k>3$ and $O\left(1.356^n\right)$ for $k=3$, where $n$ is the number of variables. 
Feder and Motwani improved upon the Beigel-Eppstein algorithm for $k\geq 11$.
Hertli, Hurbain, Millius, Moser, Scheder and Szedl{\'a}k improved these bounds for Unique $(k,2)$-CSP for every $k\geq 5$.

We improve the result of Hertli et al. and obtain better bounds for Unique~$(k,2)$-CSP for~$k\geq 5$. In particular, we improve the running time of Unique~$(5,2)$-CSP from~$O\left(2.254^n\right)$ to~$O\left(2.232^n\right)$ and Unique~$(6,2)$-CSP from~$O\left(2.652^n\right)$ to~$O\left(2.641^n\right)$.

Recently, Li and Scheder also published an improvement over the algorithm of Hertli et al. in the same regime as ours. Their improvement does not include quantitative bounds, we compare the works in the paper. 
\end{abstract}

\section{Introduction}
The general Constraint Satisfaction Problem, in which we are asked to find an assignment to a set of variables that satisfies a list of arbitrary constraints, is NP-Complete. 
Furthermore, it is widely believed that a substantial improvement over the naive exhaustive search is unlikely for the general CSP problem and even for special cases of it like The \emph{Boolean} Satisfiability Problem (SAT).
Nevertheless, when the structure of the input is restricted in a certain manner, there are known improvements in the form of \emph{moderately exponential} algorithms. These are algorithms that still have an exponential running time, yet achieve an exponential improvement over the exhaustive search bounds.

The study of moderately exponential algorithms for NP-Complete problems is extensive.
In fact, exponential yet better-than-naive algorithms for NP-Complete problems were known for some problems, for example The Travelling Salesman Problem, long before the definition of NP.
A survey of Woeginger \cite{woeginger2003exact} covers and refers to dozens of papers exploring such algorithms for many problems including satisfiability, graph coloring, knapsack, TSP, maximum independent sets and more.
Subsequent review article of Fomin and Kaski \cite{fomin2013exact} and book of Fomin and Kratsch \cite{fomin2010exact} further cover the topic of exact exponential-time algorithms.

Two of the most notable problems for which the study of moderately exponential algorithms was fruitful are $k$-satisfiability (usually abbreviated as $k$-SAT) and graph coloring.

For satisfiability, the running time of the trivial algorithm enumerating over all possible assignments is $O^*(2^n)$.
No algorithms solving SAT in time $O^*\left(\left(2-\varepsilon\right)^n\right)$ for any $\varepsilon>0$ are known, and a popular conjecture called The \emph{Strong} Exponential Time Hypothesis \cite{calabro2009complexity} states that no such algorithm exists. 
On the other hand, for every fixed $k$ there exists a constant $\varepsilon_k>0$ such that $k$-SAT (i.e., SAT on formulas in CNF form with at most~$k$ literals in every clause) can be solved in $O^*\left(\left(2-\varepsilon_k\right)^n\right)$ time.
A result of this type was first published by Monien and Speckenmeyer in 1985 \cite{monien1985solving}. 
A long list of improvements for the values of $\varepsilon_k$ were published since, including the celebrated 1998 PPSZ algorithm of Paturi, Pudl{\'a}k, Saks and Zane \cite{paturi2005improved} and the recent improvement over it by Hansen, Kaplan, Zamir and Zwick \cite{hansen2019faster}.
The PPSZ bound was originally obtained only for the case of input formulas with a unique satisfying assignment, its analysis was extended to the general case more than a decade later by Hertli~\cite{Hertli14b}.

For graph coloring, i.e, the problem of deciding whether a graph is $k$-colorable, the naive exhaustive search algorithm takes $O(k^n)$ time. 
Nevertheless, it is known that computing the chromatic (or coloring) number of a graph (i.e., the smallest~$k$ for which the graph is $k$-colorable) can be done in exponential time that does not depend on~$k$.
The first such result was an $O^*(3^n)$ algorithm of Lawler \cite{lawler1976note}. 
A long line of works followed until finally an algorithm computing the chromatic number in $O^*(2^n)$ time was devised by Bj{\"o}rklund, Husfeldt and Koivisto in 2009 \cite{bjorklund2009set}. This is conjectured to be optimal.
For~$k\leq 6$ there are algorithms solving $k$-colorability exponentially faster than~$O(2^n)$~\cite{beigel20053}\cite{mycoloring}. It is currently not known even if $7$-coloring can be solved exponentially faster than the general $O(2^n)$ bound of computing the coloring number.
One of the biggest open problems in this field is whether $k$-coloring can be solved in $O^*\left(\left(2-\varepsilon_k\right)^n\right)$ time for every fixed~$k$.

Both examples are special cases of the more general Constraint Satisfaction Problem.
In an $(a,b)$-formula, we have~$n$ variables such that each of them can take a value in~$[a]:=\{1,\ldots,a\}$ and a list of constraints such that each constraint may depend on at most~$b$ variables.
Every such constraint can be equivalently replaced by a disjunction of at most~$a^b$ constraints of the form~$\left(x_1\neq c_1 \;\vee\; x_2\neq c_2 \;\vee\; \ldots \;\vee\; x_b\neq c_b\right)$ where~$x_1,\ldots,x_b$ are (not necessarily distinct) variables and~$c_1,\ldots,c_b \in [a]$ are possible values.
Thus, we can think of every $(a,b)$-formula as a list of constraints of that form.
In the $(a,b)$-CSP problem we are given a $(a,b)$-formula and need to decide whether or not there is an assignment to the variables that satisfies all constraints. 
In the Unique $(a,b)$-CSP problem we add the assumption that if there is such an assignment it is unique.
Note that $k$-SAT is the same as $(2,k)$-CSP, and that $k$-coloring is a special case of $(k,2)$-CSP. We later elaborate on the close relation between the $(k,2)$-CSP and $k$-coloring problems.

In this paper we focus on obtaining better algorithms for Unique $(k,2)$-CSP.

\subsection{Possible running time}
Denote by~$c_{a,b}$ the infimum of constants~$c$ such that $(a,b)$-CSP on formulas with~$n$ variables can be solved in $O\left(\left(c+o(1)\right)^{n}\right)$ time.
Naively, $c_{a,b}\leq a$ as we can simply try all~$a^n$ possible assignments to the variables.
A simple improvement comes from the use of \emph{down-sampling}. 
Given a $(a,b)$-formula we may randomly restrict each variable to~$a'<a$ uniformly chosen values. Each satisfying assignment is not ruled out by the restriction with probability~$\left(\frac{a'}{a}\right)^n$. 
After this down-sampling step, we are left with a $(a',b)$-formula.
Thus, for every~$a'<a$ we have~$c_{a,b}\leq \frac{a}{a'} \cdot c_{a',b}$.
In particular, $c_{a,b} \leq \frac{a}{2} c_{2,b}$.
As we know that $k$-SAT can be solved exponentially faster than~$O(2^n)$ for every fixed~$k$, we have that~$c_{2,b}<2$ and in particular the strict inequality~$c_{a,b}<a$ holds for every~$a,b$.

On the other hand, $(a,b)$-CSP is clearly NP-Complete for every~$a>1$ except of the special case of~$(a,b)=(2,2)$ (which is the polynomial $2$-SAT).
The Exponential Time Hypothesis \cite{calabro2009complexity} states that there exists some constant~$c>0$ such that $3$-SAT takes $\Omega(2^{cn})$ time to solve.
Traxler~\cite{traxler2008time} showed that assuming the Exponential Time Hypothesis, there exists some~$c'>0$ such that~$c_{k,2}>k^{c'}$. Namely, even for~$b=2$ the $(k,2)$-CSP problem becomes strictly more complex as~$k$ increases. 

\subsection{$(k,2)$-CSP and $k$-Coloring}
Consider the following hierarchy of three problems:
\begin{itemize}
\item \textbf{$k$-Coloring:} given a graph with $n$ vertices, determine whether it is $k$-colorable.
\item \textbf{$k$-List-Coloring:} given a graph with $n$ vertices and a list of at most~$k$ allowed colors (from a possibly larger universe of colors) for each vertex, determine if there is a proper coloring of the graph such that each vertex is colored with one of the allowed colors in its list. 
\item \textbf{$(k,2)$-CSP:} given $n$ variables that can admit values from~$[k]$ and a list of arbitrary constraints involving one or two variables each, determine if there is an assignment of values to the variables that satisfies all constraints.
\end{itemize}

Each problem is a special case of the next one.
Every instance of $k$-coloring is also an instance of $k$-list-coloring where all lists are simply~$[k]$. Every instance of $k$-list-coloring is also an instance of $(k,2)$-CSP.
Nevertheless, while $k$-coloring and $k$-list-coloring can both be solved in~$O^*(2^n)$ time regardless of~$k$~\cite{bjorklund2009set}, Traxler's reduction~\cite{traxler2008time} shows that there is some constant~$k_0$ such that for every~$k>k_0$ we have~$c_{k,2} \geq 2$. Let $k_0$ be the minimal such constant.
It is currently known that~$c_{4,2}<2$ and thus~$k_0\geq 4$.

In \cite{mycoloring} it was recently shown that if $k$-list-coloring can be solved in~$O^*\left(\left(2-\varepsilon\right)^n\right)$ time for some~$\varepsilon>0$ then $\left(k+2\right)$-coloring can also be solved in~$O^*\left(\left(2-\varepsilon'\right)^n\right)$ for some~$\varepsilon'>0$.
In particular, $\left(k_0+2\right)$-coloring can be solved exponentially faster than~$O(2^n)$.
As~$k_0\geq 4$, this resulted in the first~$O^*\left(\left(2-\varepsilon\right)^n\right)$ time algorithms for $5$-coloring and $6$-coloring.
This gives a strong motivation for improving upper bounds for $(k,2)$-CSP, with the goal of improving the bound on~$k_0$.
In particular, showing that $(5,2)$-CSP can be solved in~$O^*\left(\left(2-\varepsilon\right)^n\right)$ time would result in the first~$O^*\left(\left(2-\varepsilon\right)^n\right)$ time algorithm for $7$-coloring.

\subsection{Previous results}
By the down-sampling argument we can reduce $(k,2)$-CSP to the polynomial $(2,2)$-CSP (i.e.,  $2$-SAT) and get an expected running time of $O\left(\left(\frac{k}{2}\right)^n\right)$.
Beigel and Eppstein \cite{beigel20053} gave an algorithm for $(k,2)$-CSP running in time $O\left(\left(0.4518k\right)^n\right)$ for $k>3$ and $O\left(1.356^n\right)$ for $k=3$.
Feder and Motwani~\cite{feder2002worst} give a $(k,2)$-CSP algorithm based on the $k$-SAT PPZ algorithm, which is the predecessor of the PPSZ one. They improve on the bound of Beigel and Eppstein only for~$k\geq 11$.
Hertli, Hurbain, Millius, Moser, Scheder and Szedl{\'a}k \cite{hertli2016ppsz} improved the bounds for Unique $(k,2)$-CSP for every~$k\geq 5$.
Several other works \cite{schoning1999probabilistic} \cite{li2008k} \cite{scheder2010ppz} focus on the case where~$b>2$.

\subsection{Our contribution}

\begin{figure}
\begin{center}
\begin{tabular}{ |c|c|c|c|c|c| } 
\hline
\textbf{k} & Downsampling+2SAT & PPZ FM \cite{feder2002worst} & BE \cite{beigel20053} & PPSZ  \cite{hertli2016ppsz} & \textbf{Our algorithm} \\
\hline
3 & 1.5 & 1.818 & \emph{1.365} & 1.434 & - \\ 
4 & 2 & 2.214 & \emph{1.808} & 1.849 & - \\ 
5 & 2.5 & 2.606 & 2.259 & \emph{2.254} & \textbf{2.232} \\
6 & 3 & 2.994 & 2.711 & \emph{2.652} & \textbf{2.641} \\
7 & 3.5 & 3.381 & 3.163 & \emph{3.045} & \textbf{3.042} \\
\hline
\end{tabular}
\end{center}
\caption{Comparisons of the exponent base in Unique $(k,2)$-CSP algorithms}
\label{table}
\end{figure}

We present an algorithm that improves on the result of Hertli et al. and obtain better bounds for Unique~$(k,2)$-CSP for~$k\geq 5$. In particular, we improve the running time of Unique~$(5,2)$-CSP from~$O\left(2.254^n\right)$ to~$O\left(2.232^n\right)$ and Unique~$(6,2)$-CSP from~$O\left(2.652^n\right)$ to~$O\left(2.641^n\right)$. 
Our result is compared to the previous ones in Table~\ref{table}.

We obtain our result by combining the strengths of both PPSZ and the Beigel-Eppstein algorithms.
Intuitively, we make the following insight regarding PPSZ-type algorithms.
Throughout the run of the PPSZ algorithm, it slowly manipulates the CSP formula.
For every~$k'<k$ there is some time-point such that if we stop the algorithm at that point then the formula roughly looks like a $(k',2)$-CSP formula. 
Furthermore, the rest of the algorithm run would have looked similar to running PPSZ on a $(k',2)$-formula.
Thus, for~$k\geq 5$ we may stop the run of the PPSZ algorithm when the formula looks similar to a $(4,2)$-CSP or $(3,2)$-CSP formula, and then switch to using the Beigel-Eppstein algorithm which is faster than PPSZ for~$k\leq 4$.

In Section~\ref{sec:prev} we give an extensive overview of the previous results we need to use. Then in Section~\ref{sec:ours} we introduce our algorithm and obtain the improved bounds.
In Section~\ref{sec:improve} we sketch possible improvements to the analysis of Section~\ref{sec:ours} and show that even with a completely ideal analysis of our algorithm we would improve the bound for~$k=5$ from~$O(2.232^n)$ only to~$O(2.223^n)$.
Thus, to prove that~$k_0\geq 5$, if true, additional algorithmic tools are necessary.
In Section~\ref{sec:conclude} we conclude our work, discuss more possible uses for the PPSZ-related observations, and present open problems. 

\subsection{Comparison with a recent work of Li and Scheder}
Recently, Li and Scheder \cite{li2021impatient} also published an improvement for the PPSZ-type algorithm of Hertli et al. for Unique CSP.
Their algorithm does not use the Beigel-Eppstein algorithm and just modifies the PPSZ-type algorithm itself in a different manner to ours.
They prove that this modification gives an exponential improvement over the bounds of Hertli et al., but do not give a quantitative bound of this improvement.
We could not find a way to compute such a bound from their paper, but suspect this improvement is very small.

Li and Scheder also observe that during the run of the PPSZ-type algorithm when the number of color choices of a variable is very low (in their algorithm, when it reaches~$2$), then there are better ways to settle the value of the variable than continuing the run of the PPSZ algorithm. In their case, they do so by randomly picking one of the two colors for the variable.

Our work can be seen as a refinement of this idea in two different ways. First, we cut off the PPSZ run with small color sets yet larger than two. 
Second, we resolve the remaining instance with a variant of the Beigel-Eppstein algorithm, which is much better than a random choice. 

\section{Relevant overview of previous work}\label{sec:prev}
In this section we give an overview of all previous results that are used in our algorithm. 
We repeat and refine some of the theorems used in these papers for their later use in Section~\ref{sec:ours}.

\subsection{The algorithm of Beigel and Eppstein}\label{sec:BE}
The algorithm of Beigel and Eppstein \cite{beigel20053} solves a CSP by performing a series of local reductions that either reduce the number of variables in the CSP or the number of allowed values in some of the variables.

An example for such a local reduction that is of particular interest to us follows.
\begin{lemma}\label{be2}[Lemma 2 of \cite{beigel20053}]
Let~$(V,F)$ be a CSP in which each variable~$x\in V$ has~$k(x)$ allowed values.
Let~$x$ be a variable with~$k(x)=2$, then there exists a set~$F'$ of additional constraints, each of size two, such that~$(V,F)$ is satisfiable if and only if~$(V\setminus\{x\},F\cup F')$ is satisfiable, with the same number of allowed values for each variable~$y\neq x$.
\end{lemma}

The result claimed in \cite{beigel20053} is that their algorithm solves $(3,2)$-CSPs in time~$O\left(1.3645^n\right)$ and $(k,2)$-CSPs for~$k>3$ in time~$O\left(\left(0.4518k\right)^n\right)$. 
Note that $1.3645 > 1.3554 = 0.4518\cdot 3$.
In fact, the result proved in their paper is slightly stronger than that. 
The following Theorem follows from \cite{beigel20053}.

\begin{theorem}\label{be34}[Section 5 of \cite{beigel20053}]
Let~$(V,F)$ be a CSP with~$|V|=n_3+n_4$ variables such that~$n_3$ variables have three allowed values and~$n_4$ variables have four allowed values, then we can solve it in time~$O\left(1.3645^{n_3} \cdot 1.8072^{n_4}\right)$.
\end{theorem}

The claimed results follow from Theorem~\ref{be34} immediately. 
For $(3,2)$-CSPs we simply have~$n_3=n,\;n_4=0$ and for $(k,2)$-CSPs with $k>3$ we down-sample each variable to~$4$ out of its~$k$ possible values and then use the theorem with~$n_3=0,\;n_4=n$, with a total expected run-time of~$O\left(\left(\frac{k}{4}\right)^{n} \cdot 1.8072^{n}\right) = O\left(\left(0.4518k\right)^n\right)$.
Nevertheless, for our use we need to fully use the power of Theorem~\ref{be34} and we even slightly refine it with the following statement, from now on referred to as \emph{the extended BE algorithm}.

\begin{theorem}\label{eBE}
Denote by
\[   
BE(i) := 
     \begin{cases}
       1 &\quad\text{if } i\leq 2\\
       1.3645 &\quad\text{if }i=3 \\
       0.4518\cdot i &\quad\text{if }i\geq 4
     \end{cases}
.\]
Let~$(V,F)$ be a CSP in which each variable~$x\in V$ has~$k(x)$ allowed values. 
Let~$n_i$ be the number of variables~$x$ in~$V$ such that~$k(x)=i$. 
Then, we can solve~$(V,F)$ in~$O\left(\prod_{i} BE(i)^{n_i}\right)$ expected time.
\end{theorem}
\begin{proof}
The~$n_1$ variables with a single possible value can be ignored. 
The~$n_2$ variables with two possible values can be eliminated using Lemma~\ref{be2}.
For every~$i>4$ we use down-sampling to reduce the number of allowed values to four.
We finally apply Theorem~\ref{be34}.
\end{proof}

\subsection{The PPZ and PPSZ-type algorithms}\label{sec:PPZPPSZ}
In this section we present an overview of \cite{feder2002worst} and \cite{hertli2016ppsz}.
We state theorems of both papers and their variants that are useful for our analysis, and adapt some of their notation.
Throughout the section we discuss only instances with a unique satisfying assignment.

\begin{definition} [$D$-implication]
Let $F$ be a $(k,2)$-CSP formula over a set $V$ of variables, $x\in V$ be a variable, $c\in [k]$ be a possible value, $\alpha_0$ a partial assignment, and $D\in \mathbb{N}$.
We say that $\alpha_0$ $D$-implies $x\neq c$ and write $\alpha_0 \models_D \left(x\neq c\right)$ if there is a subset of constraints $G\subseteq F$ of size $|G|\leq D$ such that $G\wedge \alpha_0$ implies $\left(x\neq c\right)$.
\end{definition}

By enumeration, we can check whether $\alpha_0 \models_D \left(x\neq c\right)$ in $O\left(|F|^D \cdot poly(n)\right)$ time, which is polynomial in $n,k$ if $D$ is a constant and sub-exponential in $n$ even if $D$ is a slow-enough growing function of $n$.
For the rest of the section we fix $D$.

\begin{definition} [Eligible values]
Let $F$ be a $(k,2)$-CSP formula over a set $V$ of variables, $\alpha_0$ a partial assignment, and $x\in V\setminus V(\alpha_0)$ a variable which value is not assigned in $\alpha_0$.
We denote by 
\[
\mathcal{A}\left(x, \alpha_0 \right) := 
\{
c\in [k] \;|\; \alpha_0 \not\models_D \left(x\neq c\right)
\}
\]
the set of all possible values for $x$ that are not ruled out by $D$-implication from $\alpha_0$.
\end{definition}

We can now describe the PPSZ algorithm (adapted from SAT \cite{PPSZ98} to CSP in \cite{hertli2016ppsz}).
Given a $(k,2)$-CSP $F$, we begin with $\alpha_0 = \emptyset$ the empty assignment and incrementally add variables to it, hoping to finish with a satisfying assignment.
In particular, we choose a permutation $\pi$ of the variables $V$ uniformly at random, and then choose an assignment for the variables of $V$ one-by-one according to the order of $\pi$.
When we reach a variable $x$, we compute $\mathcal{A}\left(x, \alpha_0 \right)$ in sub-exponential time, pick a uniformly random $c \sim U\left(\mathcal{A}\left(x, \alpha_0 \right)\right)$, and extend $\alpha_0$ by setting $\alpha_0(x)=c$.

\begin{algorithm}
\caption{The PPSZ algorithm}\label{alg:ppsz}
Pick a uniform random permutation $\pi$ of the set $V$ of variables\;
Set $\alpha_0=\emptyset$\;
\For{$x\in V$ in the order dictated by $\pi$}
{
Draw $c \sim U\left(\mathcal{A}\left(x, \alpha_0 \right)\right)$ uniformly\;
Set $\alpha_0(x) := c$\;
}
Return $\alpha_0$\;
\end{algorithm}

Algorithm~\ref{alg:ppsz} runs in sub-exponential time and returns some assignment $\alpha_0$ to all variables of $V$.
It is clear that the probability of $\alpha_0$ to satisfy $F$, assuming that $F$ is satisfiable, is at least $k^{-n}$.
We next prove that for formulas $F$ with exactly one satisfying assignment $\alpha$, the probability that Algorithm~\ref{alg:ppsz} produces the satisfying assignment $\alpha_0=\alpha$ is in fact exponentially larger. For the rest of the section we assume that $F$ has a unique satisfying assignment and denote it by $\alpha$.

\begin{definition}[Ultimately eligible values]
Let $F$ be a $(k,2)$-CSP formula uniquely satisfied by $\alpha$, $\pi$ be a permutation of its variables $V$ and $x\in V$ some variable.\\
We let $V_{\pi, x} := \{y \in V \;|\; \pi(y)<\pi(x) \}$ be the set of all variables appearing before $x$ in $\pi$, $\alpha_{\pi,x} := \restr{\alpha}{V_{\pi,x}}$ be the partial assignment resulting by restricting $\alpha$ to $V_{\pi,x}$ and then we denote by
$
\mathcal{A} \left(x, \pi \right) :=
\mathcal{A} \left(x, \alpha_{\pi,x}\right)
$
the set of all possible values for $x$ that are not ruled out by $D$-implication when we reach $x$ in a PPSZ iteration with permutation $\pi$, given that all previous variables were set correctly.
\end{definition}

We observe that Algorithm~\ref{alg:ppsz} returns $\alpha$ if and only if it draws the correct value for \emph{every} variable. In particular, for a specific permutation $\pi$ the probability of success is exactly $\prod_{x\in V} \frac{1}{|\mathcal{A} \left(x, \pi \right)|}$.
For a random permutation then, the probability of success is
\[
\E_\pi \left[\prod_{x\in V} \frac{1}{|\mathcal{A} \left(x, \pi \right)|}\right]
\geq
k ^ {-\sum_{x\in V} \E_\pi\left[\log_k |\mathcal{A} \left(x, \pi \right)|\right] }
\]
where we use Jensen's inequality.
In particular, it is enough to give an upper bound on $\E_\pi\left[\log_k |\mathcal{A} \left(x, \pi \right)|\right]$ that holds for every variable $x$.

\subsection{The PPZ-type algorithm of Feder and Motwani}\label{subsec:PPZ}
In the PPZ-type variant of Feder and Motwani \cite{feder2002worst}, a simpler variant where $D=1$ is presented and analysed.
Namely, a possible value~$c$ for a variable~$x$ is ruled out if and only if a variable~$y$ appeared before~$x$ in the permutation and was assigned a value~$c'$ such that the constraint~$\left(x\neq c \;\vee\; y\neq c'\right)$ appears in the list of constraints. 
When we reach the variable~$x$, we uniformly guess a value for it out of all values that are not ruled out in that manner.

\begin{lemma}\label{criticalvar}
Let $(V,F)$ be the sets of variables and constraints in a Unique~$(k,2)$-CSP.
Denote by~$\varphi$ the unique satisfying assignment of~$(V,F)$.
For every variable~$x\in V$ and every value~$\varphi(x)\neq c'\in [k]$ other than the value~$x$ is assigned in~$\varphi$, there exists a variable~$y=y_{x,c'}\in V\setminus\{x\}$ such that~$\left(x\neq c' \;\vee\; y\neq \varphi(y) \right)\in F$.
\end{lemma}
\begin{proof}
Assume by contradiction that there exists a variable~$x$ and a value~$c'\neq\varphi(x)$ for which~$\left(x\neq c' \;\vee\; y\neq \varphi(y) \right)\notin F$ for every variable~$y$.
Consider the assignment~$\varphi'$ such that~$\varphi'(x)=c'$ and~$\varphi'(y)=\varphi(y)$ for every~$y\neq x$. It is a satisfying assignment as well, and~$\varphi'\neq\varphi$ which contradicts the uniqueness assumption.
\end{proof}

Instead of uniformly drawing a permutation~$\pi\sim S_{|V|}$ of the variables, we (equivalently) independently draw a \emph{time} value~$\pi(x)\sim U\left(\left[0,1\right]\right)$ uniformly for every variable~$x$, and let~$\pi$ be the permutation induced by the order of the time values~$\pi(x)$ for~$x\in V$.

We observe that if~$\pi(y_{x,c'})<\pi(x)$ then~$c'\notin \mathcal{A} \left(x, \pi \right)$.
In particular, if~$\pi(x)=p\in [0,1]$ then for every~$c'\neq \varphi(x)$ with probability at least~$p$ we have that~$c'\notin \mathcal{A} \left(x, \pi \right)$.

\begin{lemma}\label{ppzforceprob}
For every variable~$x$, we have~
\[
\E_\pi\left[\log_k |\mathcal{A} \left(x, \pi \right)| \;\;\middle|\;\; \pi(x)=p\right]
\leq
\sum_{i=0}^{k-1} \binom{k-1}{i} \left(1-p\right)^{i} p^{k-1-i} \log_k \left(1+i\right)
.
\]
\end{lemma}
\begin{proof}
For the analysis, we may assume that we rule out values only by the constraints involving~$x$ and one of the variables~$y_{x,c'}$ for~$c'\neq \varphi(x)$. This holds since ruling out more variables can only decrease the size of~$\mathcal{A} \left(x, \pi \right)$.

If the variables~$y_{x,c'}$ are distinct for all~$c'\neq \varphi(x)$, then the right hand side of the lemma's statement is exactly the expected size of~$\mathcal{A} \left(x, \pi \right)$, conditioned on~$\pi(x)=p$. 
This is simply the expectation of~$\log_k \left(1+i\right)$ where~$i\sim Binomial\left(k-1,1-p\right)$ is a binomial random variable. 

Generally, let~$A_{c'}$ be the indicator for the event that~$\pi(y_{x,c'})>p$. 
We need to upper bound $\E[\log_k \left(1 + \sum_{c'\neq \varphi(x)} A_{c'}\right)]$.
Let~$A'_{c'}$ be independent Bernoulli random variables with probability~$(1-p)$ to be~$1$ and probability~$p$ to be~$0$.
By concavity of the function~$\log_k(1+z)$ and Jensen's inequality it follows that 
$\E[\log_k \left(1 + \sum_{c'\neq \varphi(x)} A_{c'}\right)] \leq
\E[\log_k \left(1 + \sum_{c'\neq \varphi(x)} A'_{c'}\right)]$, which concludes our proof.
The complete proof of the last statement appears in~\cite{hertli2016ppsz} as Lemma~$A.1$.
\end{proof}

Denote by~$S'_{k,2} := \int_{0}^{1} \sum_{i=0}^{k-1} \binom{k-1}{i} \left(1-p\right)^{i} p^{k-1-i} \log_k \left(1+i\right) dp$.
By Lemma~\ref{ppzforceprob}, $\E \left[\log_k |\mathcal{A} \left(x, \pi \right)|\right] = \int_0^1 \E_\pi\left[\log_k |\mathcal{A} \left(x, \pi \right)| \;\;\middle|\;\; \pi(x)=p\right] dp \leq S'_{k,2}$, this concludes the analysis of the Feder-Motwani PPZ-type algorithm.

\begin{theorem}[\cite{feder2002worst}]\label{ppzcorrect}
The success probability of a PPZ iteration is at least~$k^{-S'_{k,2}}$.
\end{theorem}

\subsection{The PPSZ-type algorithm of Hertli et al.}\label{subsec:PPSZ}

In the PPSZ algorithm analysed in \cite{hertli2016ppsz} more involved~$D$-implications are considered.
In the analysis for~$D=1$, we noticed that for every variable~$x$ and every value~$c'\neq \varphi(x)$ there exists some variable~$y_{x,c'}$ such that if~$\pi(y_{x,c'})<\pi(x)$ then~$c'\notin \mathcal{A} \left(x, \pi \right)$. 

We say that a variable~$y$ is \emph{decided} with respect to some partial assignment~$\alpha_0$ if~$|\mathcal{A} \left(y, \alpha_0 \right)|=1$, i.e., if~$\alpha_0$ already $D$-implies the correct value of~$y$ in~$\varphi$.
The main observation is that if in time~$p:=\pi(x)$ the variable~$y_{x,c'}$ is decided then~$c'\notin \mathcal{A} \left(x, \pi \right)$.
The variable~$y_{x,c'}$ is necessarily decided if~$\pi(y_{x,c'})<p$ but can also be decided if it is yet to appear in the permutation. 
Thus, the probability of~$y_{x,c'}$ being decided at time~$p$ is strictly larger than~$p$.

We give an intuitive reasoning for the probability of a variable being decided.
Denote by~$q_k(p)$ the probability that a variable~$x$ is decided by time~$p$. 
The variable~$x$ is decided by time~$p$ if~$\pi(x)<p$ or alternatively if for every~$c'\neq \varphi(x)$ the variable~$y_{x,c'}$ is by itself decided at time~$p$. 
In particular, $q_k(p)$ is a solution to the recurrence~$q_k(p) = p + (1-p)q_k(p)^{k-1}$.
We thus denote by~$q_k(p)$ the smallest non-negative real solution to that recurrence, it can be analytically computed for every~$k$ as it is simply a root of a polynomial.

This intuitive argument is of course not complete and lacks many technical details. 
Nevertheless, this statement does hold, and the following strengthening of Lemma~\ref{ppzforceprob} and Theorem~\ref{ppzcorrect} are proven in~\cite{hertli2016ppsz}.

\begin{lemma}\label{ppszforceprob}[$A.1$ in \cite{hertli2016ppsz}]
For every variable~$x$, we have~
\[
\E_\pi\left[\log_k |\mathcal{A} \left(x, \pi \right)| \;\;\middle|\;\; \pi(x)=p\right]
\leq
\sum_{i=0}^{k-1} \binom{k-1}{i} \left(1-q_k\left(p\right)\right)^{i} q_k\left(p\right)^{k-1-i} \log_k \left(1+i\right)
.
\]
\end{lemma}

Denote by $S_{k,2} := \int_0^1 \sum_{i=0}^{k-1} \binom{k-1}{i} \left(1-q_k\left(p\right)\right)^{i} q_k\left(p\right)^{k-1-i} \log_k \left(1+i\right) dp$.

\begin{theorem}[Correctness of \cite{hertli2016ppsz}]
Let $F$ be a Unique $(k,2)$-CSP formula, then for every variable $x$ it holds that
$\E_\pi\left[\log_k |\mathcal{A} \left(x, \pi \right)|\right]
\leq S_{k,2} + \varepsilon_D$,
where $\varepsilon_D$ is some error parameter that depends only on $D$ and goes to $0$ as $D$ goes to infinity.
\end{theorem}

\section{Faster Unique $(k,2)$-CSP algorithm}\label{sec:ours}
On a very high-level, our algorithm combines Hertli et al.'s PPSZ (Section~\ref{sec:PPZPPSZ}) with the BE algorithm (Section~\ref{sec:BE}).
We begin by illustrating our idea intuitively (initially ignoring some crucial technical details to be discussed later).
Consider a run of the PPSZ algorithm, as described in Section~\ref{sec:PPZPPSZ}.
For the early variables in the permutation $\pi$, it is very likely that $|\mathcal{A} \left(x, \pi \right)| = k$, since $\alpha_0$ assigns values to very few variables.
On the other hand, for the last variables in the permutation, it is very likely that $|\mathcal{A} \left(x, \pi \right)| = 1$.
It turns out that in any point throughout the run of a PPSZ iteration, the sizes $|\mathcal{A} \left(x, \alpha_0 \right)|$ for the remaining variables $x\in V\setminus V(\alpha_0)$ are quite concentrated. 
Furthermore, after most of the variables have~$|\mathcal{A} \left(x, \alpha_0 \right)| \approx k' < k$ the remaining portion of the PPSZ iteration strongly resembles a PPSZ algorithm for~$(k',2)$-CSP formulas.
As we see in Table~\ref{table}, for~$k<5$ PPSZ behaves worse on $(k,2)$-CSP formulas than the BE algorithm.

Thus, in our algorithm, we begin with an iteration of PPSZ but halt it somewhere in the middle of the permutation when the sizes $|\mathcal{A} \left(x, \alpha_0 \right)|$ are concentrated in~$1,2,3,4$. 
At that point, we use the extended BE algorithm shown in Section~\ref{sec:BE}.

We set some parameter~$t\in[0,1]$ to be chosen later and consider the following algorithm.
\begin{algorithm}[H]
\caption{Our algorithm}\label{alg}
Pick a uniform random permutation $\pi$ of the set $V$ of variables\;
Denote by $\pi_{<t}$ the prefix of~$\pi$ of size $t|V|$ and by~$V_{<t}$ the variables appearing in it\;
Set $\varphi=\emptyset$\;
\For{$x\in V_{<t}$ in the order dictated by $\pi_{<t}$}
{
Draw $c \sim U\left(\mathcal{A}\left(x, \varphi \right)\right)$ uniformly\;
Set $\varphi(x) := c$\;
}
Run the extended BE algorithm on the remaining CSP $F$\;
Return the solution $\varphi$\;
\end{algorithm}

Consider an iteration of Algorithm~\ref{alg}.
Denote by~$R_i$ the number of variables that appeared in~$V_{<t}$ and had $|\mathcal{A}\left(x, \pi \right)| = i$.
Denote by~$\varphi'$ the partial assignment constructed by time~$t$.
Let~$B_i$ be the number of variables that did not appear in~$V_{<t}$ and had $|\mathcal{A}\left(x, \varphi' \right)| = i$.

\begin{lemma}
The success probability of Algorithm~\ref{alg} is $\prod_{i=1}^{k} i^{-R_i} \cdot \prod_{i=5}^{k} \left(\frac{4}{i}\right)^{B_i}$.
\end{lemma}
\begin{proof}
The probability of all PPSZ assignments to be correct is~$\prod_{i=1}^{k} i^{-R_i}$, as follows from Section~\ref{sec:PPZPPSZ}.
The probability of the random down-sampling to not rule out the correct assignments is $\prod_{i=5}^{k} \left(\frac{4}{i}\right)^{B_i}$.
\end{proof}

\begin{lemma}
The running time of Algorithm~\ref{alg} is $O\left(1.3645^{B_3} \cdot 1.8072^{B_4 + \ldots + B_k}\right)$.
\end{lemma}
\begin{proof}
The running time of the (partial) PPSZ iteration is polynomial. 
The running time of the BE algorithm is $O(1.3645^{n_3} \cdot 1.8072^{n_4})$.
\end{proof}

Note that all~$R_i$ and~$B_i$ are fully determined by the choice of~$\pi$.
Thus, for a specific choice of~$\pi$, if we repeatedly run Algorithm~\ref{alg} with~$\pi$ we expect finding a solution after~$\left(\prod_{i=1}^{k} i^{-R_i} \cdot \prod_{i=5}^{k} \left(\frac{4}{i}\right)^{B_i}\right)^{-1}$ iterations.
In particular, after \[\left(\prod_{i=1}^{k} i^{-R_i} \cdot \prod_{i=5}^{k} \left(\frac{4}{i}\right)^{B_i}\right)^{-1} \cdot O\left(1.3645^{B_3} \cdot 1.8072^{B_4 + \ldots + B_k}\right)
=
\prod_{i=1}^{k} i^{R_i} \cdot \prod_{i} BE\left(i\right)^{B_i}
\] computational steps.
At this point, we would like to bound the expected running time when picking a random~$\pi$ with $\prod_{i=1}^{k} i^{\E[R_i]} \cdot \prod_{i} BE\left(i\right)^{\E[B_i]} $.
Unfortunately, as we consider the running time and not a success probability (as in the PPSZ algorithm), we need an inequality of the opposite direction to Jensen's inequality.
Fortunately, this inequality \emph{essentially still holds} in this case.

\begin{lemma}[Wrong direction Jensen's inequality is still kind-of right]
Let~$\mathcal{A}$ be an algorithm with expected running time~$2^X$ conditioned on the value of a random variable~$X$.
There exists an algorithm~$\mathcal{A'}$ that successfully executes~$\mathcal{A}$ with probability at least~$0.99$ and has an expected running time of~$O\left(2^{\E\left[X\right]}\cdot \E\left[X\right]\right)$.
\end{lemma}
\begin{proof}
We apply Markov's inequality twice. First, to observe that
$$
\Pr\left(X > \E[X]+1\right)\leq \frac{1}{1+\frac{1}{\E[X]}} = 1-\frac{1}{\E[X]+1}.
$$
Hence, if we run~$\mathcal{A}$ independently for~$6\left(\E\left[X\right]+1\right)$ times, then with probability at least~$1-e^{-6} > 1-\frac{1}{200}$ at least one of these runs has~$X\leq \E\left[X\right]+1$.
Second, conditioned on any value of~$X$, with probability at least~$1-\frac{1}{200}$ algorithm~$\mathcal{A}$ finishes in less than~$200\cdot 2^{X}$ computational steps.
Thus, by union bound, if we run algorithm~$\mathcal{A}$ for~$6\left(\E\left[X\right]+1\right)$ times, and terminate each run after~$400\cdot 2^{\E[X]}$ computational steps, then at least one run of~$\mathcal{A}$ finishes with probability at least~$0.99$.
\end{proof}

\begin{corollary}
We find a satisfying assignment with probability greater than~$0.99$ in time
\[
O^\star \left(\prod_{i=1}^{k} i^{\E[R_i]} \cdot \prod_{i} BE\left(i\right)^{\E[B_i]}\right)
.
\tag{$\star$}
\]
\end{corollary}

We give a simpler analysis leading to slightly sub-optimal bounds. In Section~\ref{sec:improve} we sketch the possible improvements to the analysis presented here, and also present a clear limit to the improvements that can be achieved by this algorithm.

We slightly abuse notation by equating the numbers $R_i,B_i$ with the sets of variables they are counting.
Let~$x$ be a variable.
For the analysis we can assume that the algorithm rules out values for~$x$ only due to the constraints involving~$x$ and some~$y_{x,c'}$. This holds as ruling out more values can only improve the success probability of each iteration.

We first consider the case in which the variables~$y_{x,c'}$ for every~$c'\neq \varphi(x)$ are all distinct.

\begin{lemma}\label{partialppsz}
For each variable~$x$, we have 
\[
\E\left[\sum_{i=1}^{k} Pr\left(x\in R_i\right)\cdot \log_k i \right]
\leq
\int_0^t \sum_{i=0}^{k-1} \binom{k-1}{i} \left(1-q_k\left(p\right)\right)^{i} q_k\left(p\right)^{k-1-i} \log_k \left(1+i\right) dp
.\]
\end{lemma}
\begin{proof}
This follows immediately from Lemma~\ref{ppszforceprob}.
\end{proof}

\begin{lemma}\label{distinctvars}
Let~$x$ be a variable for which~$y_{x,c'}$ are distinct for all~$c'\neq \varphi(x)$. Then,
\[
\E\left[\sum_{i=3}^k Pr\left(x\in B_i\right) \cdot \log_k EB\left(i\right)\right]
\leq
(1-t) \cdot \sum_{i=3}^k \binom{k-1}{i} \left(1-t\right)^{i} t^{k-1-i} \cdot \log_k EB\left(i\right)
.\]
\end{lemma}
\begin{proof}
For simplicity, we analyse this part with a PPZ-type analysis (rather than PPSZ-type one), this is further discussed in Section~\ref{sec:improve}.
The probability that~$x\notin V_{<t}$ is $(1-t)$, and the probability that exactly~$i$ out of the~$(k-1)$ variables~$y_{x,c'}$ do not appear in~$V_{<t}$ is $\binom{k-1}{i} \left(1-t\right)^{i} t^{k-1-i}$.
\end{proof}

Denote by 
\begin{align*}
\text{cost}(k,t) :=&
\int_0^t \sum_{i=0}^{k-1} \binom{k-1}{i} \left(1-q_k\left(p\right)\right)^{i} q_k\left(p\right)^{k-1-i} \log_k \left(1+i\right) dp
\\&+
(1-t) \cdot \sum_{i=3}^k \binom{k-1}{i} \left(1-t\right)^{i} t^{k-1-i} \cdot \log_k EB\left(i\right)
.
\end{align*}

If all variables had completely distinct~$y_{x,c'}$'s, then by Lemma~\ref{partialppsz} and Lemma~\ref{distinctvars} we would have that $(\star)$ and in particular the running time of our algorithm is bounded by~$O\left(k^{\text{cost}(k,t)n}\right)$ for any choice of~$t$.
This would give us~$O(2.22936^n)$ for~$k=5,t=0.23$ and~$O(2.64001^n)$ for~$k=6,t=0.35$.
We now deal with the case in which these are not distinct.

In the proof of Lemma~\ref{ppzforceprob} we faced the same problem and solved it by a simple application of Jensen's inequality to the concave function~$\log_k(1+i)$. 
Unfortunately, the function $\log_k EB\left(i\right)$ is not concave (for~$i=1,\ldots,k$) due to its values on~$i=1,2$.
Indeed, the left-hand side of the inequality in Lemma~\ref{distinctvars} is higher than the right-hand side if these variables are not distinct. 
On the other hand, when these variables are not distinct then the term of Lemma~\ref{partialppsz} is much smaller.

Consider the expression
\[
\E\left[
\sum_{i=1}^{k} Pr\left(x\in R_i\right)\cdot \log_k i
+
\sum_{i=3}^k Pr\left(x\in B_i\right) \cdot \log_k EB\left(i\right)
\right]
\tag{$\star \star$}
\]
again. 
This time, we will assume that the variables~$y_{x,c'}$ are not all distinct.
Denote by~$k':=|\{y_{x,c'}\;\;|\;\;c'\neq \varphi(x)\}|<k-1$ the number of such distinct variables, and by $j_1,\ldots,j_{k'}$ their cardinalities (note that $\sum_{i=1}^{k'} j_i = k-1$).
Consider the following expression.
\begin{align*}
\E \Big[ &
\int_0^t \sum_{b_1,\ldots,b_{k'} \in \{0,1\}} q_k\left(p\right)^{k' - \sum_{i=1}^{k'} b_i} \cdot \left(1-q_k\left(p\right)\right)^{\sum_{i=1}^{k'} b_i} \cdot \log_k \left(1+\sum_{i=1}^{k'}j_i b_i\right) dp
\tag{$\star \star \star$} \\& 
+ (1-t)\cdot 
\sum_{b_1,\ldots,b_{k'} \in \{0,1\}} t^{k' - \sum_{i=1}^{k'} b_i} \cdot (1-t)^{\sum_{i=1}^{k'} b_i} \cdot  \log_k EB\left(1+\sum_{i=1}^{k'}j_i b_i\right) 
\Big].
\end{align*}
Expression $(\star \star \star)$ is a generalized form of~$\text{cost}(k,t)$ and thus upper bounds $(\star \star)$ by the same arguments.
Completely analysing the behaviour of Expression~$(\star \star \star)$ for different partitions is rather technically involved and thus we simply enumerate over the few possible cases (for small values of~$k$). 
In Section~\ref{sec:improve} we further discuss the possible improvements to the analysis of this section.

\begin{theorem}\label{bound6}
We solve Unique~$(6,2)$-CSP in $O\left(2.641^n\right)$ time.
\end{theorem}
\begin{proof}
For the choice $t=0.37$ we have that $\text{cost}(6,0.35)=\log_6(2.64001)$, this choice of~$t$ minimizes~$\text{cost}(6,t)$.
We verify that for~$t=0.35$ Expression~$(\star \star \star)$ is always lower than~$\text{cost}(6,0.35)$ and thus finish, as this implies that for every variable Expression~$(\star \star)$ is bounded by $\text{cost}(6,0.35)$.
\begin{itemize}
\item For the partition $(j_1,j_2,j_3,j_4)=(2,1,1,1)$ the value of $(\star \star \star)$ in $t=0.35$ is $\log_6(2.62023)$.
\item For the partition $(j_1,j_2,j_3)=(2,2,1)$ the value of $(\star \star \star)$ in $t=0.35$ is $\log_6(2.61171)$.
\item For the partition $(j_1,j_2,j_3)=(3,1,1)$ the value of $(\star \star \star)$ in $t=0.35$ is $\log_6(2.58391)$.
\item For the partition $(j_1,j_2)=(3,2)$ the value of $(\star \star \star)$ in $t=0.35$ is $\log_6(2.60366)$.
\item For the partition $(j_1,j_2)=(4,1)$ the value of $(\star \star \star)$ in $t=0.35$ is $\log_6(2.54819)$.
\item For the partition $(j_1)=(5)$ the value of $(\star \star \star)$ in $t=0.35$ is $\log_6(2.55566)$.
\end{itemize}
\end{proof}

\begin{theorem}\label{bound5}
We solve Unique~$(5,2)$-CSP in $O\left(2.232^n\right)$ time.
\end{theorem}
\begin{proof}
For the choice $t=0.23$ we have that $\text{cost}(5,0.23)=\log_5(2.22936)$, this choice of~$t$ minimizes~$\text{cost}(5,t)$.
This time, unfortunately, for~$t=0.23$ Expression~$(\star \star \star)$ is not always lower than~$\text{cost}(5,0.23)$. In particular, it is for every partition except of $(j_1)=4$.
\begin{itemize}
\item For the partition $(j_1,j_2,j_3)=(2,1,1)$ the value of $(\star \star \star)$ in $t=0.23$ is $\log_5(2.21658)$.
\item For the partition $(j_1,j_2)=(2,2)$ the value of $(\star \star \star)$ in $t=0.23$ is $\log_5(2.21983)$.
\item For the partition $(j_1,j_2)=(3,1)$ the value of $(\star \star \star)$ in $t=0.23$ is $\log_5(2.20499)$.
\item For the partition $(j_1)=(4)$ the value of $(\star \star \star)$ in $t=0.23$ is $\log_5(2.24925)$.
\end{itemize}

Denote by~$\alpha$ the fraction of variables for which the~$y_{x,c'}$ variables are all the same (i.e., variables with the only problematic partition).
With the choice~$t=0.23$ the running time of our algorithm on a formula is~$O\left(\left(2.22936^{\left(1-\alpha\right)}\cdot 2.24925^\alpha\right)^n\right)$.
On the other hand, if~$\alpha$ is large we can simply run the regular PPSZ algorithm and gain much. We see that by setting~$t=1$ and computing Expression~$(\star \star \star)$ for the same partition~$(j_1)=(4)$, gives a value of~$\log_5(2.01077)$.
Thus, running regular PPSZ would give us a running time of~$O\left(\left(2.25303^{\left(1-\alpha\right)}\cdot 2.01077^\alpha\right)^n\right)$.
We can therefore try both options ($t=0.23$ or $t=1$) simultaneously and thus get the running time of the faster one.
Both expressions balance at~$\alpha=0.08612$, giving us a running time of~$O\left(2.23107^n\right)$.
\end{proof}

Computation identical to this of Theorem~\ref{bound6} gives running time of $O\left(3.042^n\right)$ for~$k=7$ with~$t=0.44$.

\section{Improvements and Limitations}\label{sec:improve}
In this section we sketch a possible improvement to the analysis of Section~\ref{sec:ours}. The purpose of this section is not to tighten the upper bound but to explain the limitations of our algorithm and to convince that even with a tight analysis of Algorithm~\ref{alg} it will achieve running times that are only slightly better than these we get in Section~\ref{sec:ours}.
In particular, if getting~$O\left(\left(2-\varepsilon\right)^n\right)$ time for $(5,2)$-CSP is possible, new algorithmic tools are likely required.

Consider the expression $\prod_{i=1}^{k} i^{\E[R_i]} \cdot \prod_{i} EB(i)^{\E[B_i]}$ $(\star)$ proven in Section~\ref{sec:ours} to upper bound the running time of our algorithm.
In Lemma~\ref{partialppsz} we give a likely tight bound for the term involving~$\E[R_i]$, yet in Lemma~\ref{distinctvars} we settle for a PPZ-type bound for~$\E[B_i]$ in which we consider only the events in which the variables~$y_{x,c'}$ themselves appear before time~$t$ and not the events in which they are decided by that time.
The reason for this discrepancy becomes clear in the rest of the analysis. 
Due to the non-concave objective function in~$i$, we can no longer assume independence between the events of each~$y_{x,c'}$ being decided.
Nevertheless, later in Theorem~\ref{bound6} and Theorem~\ref{bound5} we observe that while the term involving~$\E[B_i]$ indeed gets worse with dependencies, the other term involving~$\E[B_i]$ gets significantly better with them and thus can cover for those dependencies. 
Ideally, then, the simple PPZ-type bound of~$t$ in Lemma~\ref{distinctvars} can be replaced with the PPSZ-type bound of~$q_k(t)$ in the total bound.
This would result in the following tighter cost function.

\begin{align*}
\tilde{\text{cost}}(k,t) :=&
\int_0^t \sum_{i=0}^{k-1} \binom{k-1}{i} \left(1-q_k\left(p\right)\right)^{i} q_k\left(p\right)^{k-1-i} \log_k \left(1+i\right) dp
\\&+
(1-t) \cdot \sum_{i=3}^k \binom{k-1}{i} \left(1-q_k\left(t\right)\right)^{i} q_k\left(t\right)^{k-1-i} \cdot \log_k EB\left(i\right)
.
\end{align*}

With this ideal cost function, we would get running times of~$O\left(2.223^n\right)$ for Unique $(5,2)$-CSP (for $t=0.32$) and~$O\left(2.628^n\right)$ for Unique $(6,2)$-CSP (for $t=0.46$).

\section{Conclusions and Open Problems}\label{sec:conclude}
In Section~\ref{sec:ours} we presented an algorithm for Unique $(k,2)$-CSP with a running time of~$O(2.232^n)$ for~$k=5$.
In Section~\ref{sec:improve} we argued that even with an ideal analysis, the bound we get for~$k=5$ is only the slightly better~$O(2.223^n)$.
Thus, it remains open and would likely require new algorithmic tools to show that $(5,2)$-CSP can be solved in~$O\left(\left(2-\varepsilon\right)^n\right)$ time, or alternatively to rule out the existence of such algorithm by reductions to popular conjectures.
More generally, we raise the following open problem.
\begin{open*}
What is the maximal~$k$ such that $(k,2)$-CSP can be solved in~$O\left(\left(2-\varepsilon\right)^n\right)$ time?
\end{open*}

The main algorithmic observation in this paper is in fact a general insight regarding the behaviour of PPSZ-type algorithms.
The Beigel-Eppstein algorithm \cite{beigel20053} only works for $(k,2)$-CSP. 
On the other hand, the PPSZ-type algorithm \cite{hertli2016ppsz} generalizes to~$b>2$ and is in fact currently the fastest algorithm for Unique $(a,b)$-CSP with~$b>2$ and any~$a$. 
Using the tools we introduced in this paper, it should be possible to turn any faster algorithm for $(a,b)$-CSP for a specific~$(a,b)$ into a faster $(a',b)$-CSP algorithm for all~$a'>a$.

Another follow-up question is whether our algorithm can be generalized to the non-unique $(k,2)$-CSP case.

%%
%% Bibliography
%%

%% Please use bibtex, 

\bibliography{csp_bib}

\appendix

\end{document}